\documentclass[sigconf, nonacm]{acmart}

\AtBeginDocument{%
  \providecommand\BibTeX{{%
    \normalfont B\kern-0.5em{\scshape i\kern-0.25em b}\kern-0.8em\TeX}}}

\setcopyright{acmcopyright}
\copyrightyear{2018}
\acmYear{2018}
\acmDOI{10.1145/1122445.1122456}

\acmConference[Woodstock '18]{Woodstock '18: ACM Symposium on Neural
  Gaze Detection}{June 03--05, 2018}{Woodstock, NY}
\acmBooktitle{Woodstock '18: ACM Symposium on Neural Gaze Detection,
  June 03--05, 2018, Woodstock, NY}
\acmPrice{15.00}
\acmISBN{978-1-4503-XXXX-X/18/06}

\usepackage{lipsum,adjustbox}

\usepackage{caption}
\usepackage{tabularx}
\usepackage{makecell}
\usepackage{graphicx}
\usepackage{tikz}
\usepackage{lipsum}
\usepackage{xcolor}
\usepackage{amsopn}
\usepackage{amsfonts}
\usepackage{csvsimple}

\usepackage{stackrel}
\usepackage{todonotes}
\usepackage{mathtools}

\usepackage{algorithm}
\usepackage[noend]{algpseudocode}

\usepackage{amsmath}

\usepackage{amssymb}
\usepackage[capitalise]{cleveref}
\usepackage{subcaption}

\usepackage{flushend}

\makeatletter
\renewcommand{\ALG@name}{\sc Algorithm}
\makeatother

\usetikzlibrary{positioning, shapes.geometric}
\usetikzlibrary{shadows}
\usetikzlibrary{positioning,arrows,shapes.arrows}

\usepackage{xspace}
\newcommand{\probsty}[1]{\textsc{#1}\xspace}
\newcommand{\algsty}[1]{\texttt{#1}\xspace}

\usepackage{xargs}
\usepackage{xifthen}
\usepackage{framed}
\usetikzlibrary{calc}

\newenvironment{tightcenter}
 {\parskip=0pt\par\nopagebreak\centering}
 {\par\noindent\ignorespacesafterend}

\usepackage{ctable}
\newlength{\RoundedBoxWidth}
\newsavebox{\GrayRoundedBox}
\newenvironment{GrayBox}[1]%
{\setlength{\RoundedBoxWidth}{\linewidth-4.5ex}
\def\boxheading{#1}
\begin{lrbox}{\GrayRoundedBox}
\begin{minipage}{\RoundedBoxWidth}%
}{%
\end{minipage}
\end{lrbox}%
\begin{tightcenter}%
\begin{tikzpicture}%
\node(Text)[draw=black!20,fill=white,rounded corners,%
inner sep=2ex,text width=\RoundedBoxWidth]%
{\usebox{\GrayRoundedBox}};
\coordinate(x) at (current bounding box.north west);
\node [draw=white,rectangle,inner sep=3pt,anchor=north west,fill=white]
at ($(x)+(10.5pt,.75em)$) {\boxheading};
\end{tikzpicture}
\end{tightcenter}\vspace{0pt}%
\ignorespacesafterend
}

\newenvironment{problem}[1]{\noindent\ignorespaces%
\FrameSep=6pt%
\parindent=0pt%
\begin{GrayBox}{\textsc{#1}}%
\newcommand\Input{Input:}%
\newcommand\Prob{Output:}%
\begin{tabular*}{\columnwidth}{@{\hspace{.25em}} >{\itshape} p{1.1cm} p{0.8\columnwidth} @{}}%
}{
\end{tabular*}%
\end{GrayBox}%
\ignorespacesafterend
}

\newtheorem{krr}{\KRfull\xspace}[]

\newtheorem{innercustomthm}{\KRfull\xspace}
\newenvironment{customkrr}[1]
  {\renewcommand\theinnercustomthm{#1}\innercustomthm}
  {\endinnercustomthm}

\def\AEWCDfull{\probsty{Annotated EWCD}}
\def\AEWCD{AEWCD\xspace}
\def\EWCD{EWCD\xspace}
\def\EWCDfull{\probsty{Exact Weighted Clique Decomposition}}
\def\WECP{WECP\xspace}
\def\WECPfull{\probsty{Weighted Edge Clique Partition}}
\def\AWECPfull{\probsty{Annotated WECP}}
\def\AWECP{\probsty{AWECP\xspace}}

\def\WBSDDWfull{\probsty{Binary Symmetric Weighted Decomposition with Diagonal Wildcards}}
\def\WBSDDWprob{\probsty{BSWD-DW}}

\def\WBSDDW{BSWD-DW\xspace}

\def\KRfull{{\sc Kernel Rule}\xspace}
\def\KR{{\sc K-rule}\xspace}

\def\KRs{{\sc K-rule}s\xspace}
\def\SRfull{{\sc Search Rule}\xspace}
\def\SR{{\sc S-rule}\xspace}
\def\SRsfull{{\sc Search Rules}\xspace}
\def\SRs{{\sc S-rule}s\xspace}

\def\KC{$K_C$\xspace}
\def\SC{$S_C$\xspace}
\def\SCS{$S_{C^*}$\xspace}
\def\KD{$K_D$\xspace}
\def\SD{$S_D$\xspace}

\def\FillNB{\algsty{FillNonBasis}}
\def\Comp{\algsty{iWCompatible}}
\def\InferCWLp{\algsty{InferCliqWts-LP}}

\def\AlgLp{\algsty{CliqueDecomp-LP}}

\def\cricca{\algsty{cricca}}
\def\criccas{\algsty{cricca*}}
\def\decaf{\algsty{DeCAF}}

\def\cw{\gamma}

\def\es{\stackrel{\star}{=}}
\def\basis{\widetilde{B}}

\def\R{\mathbb{R}}
\def\bin{\left\{ 0,1 \right\}}

\newcommand{\notes}[1]{{}} 

\begin{document}

\title{Faster Decomposition of Weighted Graphs into Cliques using Fisher's Inequality}

\author{Shweta Jain}
\affiliation{\institution{University of Utah}
  \city{Salt Lake City}
  \state{Utah}
  \country{USA}}
\email{shweta.jain@utah.edu}

\author{Yosuke Mizutani}
\affiliation{\institution{University of Utah}
  \city{Salt Lake City}
  \state{Utah}
  \country{USA}}
\email{yos@cs.utah.edu}

\author{Blair Sullivan}
\affiliation{\institution{University of Utah}
  \city{Salt Lake City}
  \state{Utah}
  \country{USA}}
\email{sullivan@cs.utah.edu}

\begin{abstract}
Mining groups of genes that consistently co-express is an important problem in biomedical research, where it is  critical for applications such as drug-repositioning and designing new disease treatments. Recently, Cooley et al. modeled this problem as \EWCDfull (\EWCD) in which, given an edge-weighted graph $G$ and a positive integer $k$, the goal is to decompose $G$ into at most $k$ (overlapping) weighted cliques so that an edge's weight is exactly equal to the sum of weights for cliques it participates in.
They show \EWCD is fixed-parameter-tractable, giving a $4^k$-kernel alongside a backtracking algorithm (together called \cricca) to
 iteratively build a decomposition. Unfortunately, because of inherent exponential growth in the space of potential solutions, \cricca is typically able to decompose graphs only when $k \leq 11$.

In this work, we establish reduction rules that exponentially decrease the size of the kernel (from $4^k$ to $k2^k$) for \EWCD. In addition, we use insights about the structure of potential solutions to give new search rules that speed up the decomposition algorithm. At the core of our techniques is a result from combinatorial design theory called Fisher's inequality characterizing set systems with restricted intersections. We deploy our kernelization and decomposition algorithms (together called \decaf) on a corpus of biologically-inspired data and obtain over two orders of magnitude speed-up over \cricca. As a result, \decaf scales to instances with $k \geq 17$.

\notes{
\begin{enumerate}
\item Decomposition of graphs into cliques is an important problem.
\item Describe the problem in 1 line
\item Describe in one line what the existing kernelization technique achieves
\item Describe what the decomposition algorithm achieves.
\item Mention what we provide: a smaller kernel and new rules for the algorithm that improve the performance both in theory and in practice.
\item Mention how much improvement in run time we are able to achieve. If we have numbers for scalability we can give those optionally.
\item Size: 2 small paragraphs or 1 big paragraph
\end{enumerate}
}
\end{abstract}

\maketitle

\ccsdesc[500]{Theory of computation~Graph algorithms analysis}
\ccsdesc[500]{Mathematics of computing~Discrete mathematics}

\keywords{cliques, decomposition, weighted graphs, Fisher's inequality, gene modules, co-expression}

\notes{

\begin{enumerate}
\item Describe the biological setup: where the problem arises in practice and what a clique decomposition means in that setup. (this can overlap significantly with the Introduction and Motivating Biological Problem sections of Madi's paper)
\item It would be nice if we can give another setup in which this problem arises, esp. since this is a data science venue and not a bio venue.
\item mention that the problem is NP-Hard
\item in a few sentences mention what we provide. Give a box plot, n vs time (in sec) comparing our algorithm with Madi's. Maybe we should think of a name for the algorithm?
\item like Madi's paper, in a small paragraph give an outline of the paper
\item Size: 1.5 columns (so page 1 ends here)

\end{enumerate}
}

\notes{
    Color guide:
    \begin{itemize}
    \item orange: Shweta's comments
    \item purple: Notes about what I intend to write
    \item black: Actual paper content
    \end{itemize}
}

\section{Introduction}

Network analysis has proven to be a very effective tool in biomedical research, in which phenomena such as the interactions between proteins and genes find natural representation as graphs~\cite{ColladoTorres2017,Lachmann2018,Venkatesan2009,Greene2015}. In gene co-expression analysis, for example, vertices represent genes and edges represent pairwise correlation between genes. Scientists are often interested in finding groups (modules) of genes that consistently co-act, which manifest as dense subgraphs or cliques in co-expression networks. The discovery of such modules is critical in understanding disease mechanisms and in the development of new therapies for diseases, especially when the primary genes associated with a disease may not be amenable to drugs~\cite{Leeuw2015,Nelson2015-ag,yan2007graph,Menche2015,Dozmorov2020}.

A recent line of work ~\cite{cooley2021,feldmann2020} models the module identification problem as \EWCDfull (\EWCD). In this problem, we are given a positive integer $k$ and a graph $G$ whose edges have positive weights. The goal is to find a decomposition of the vertices of the graph into at most $k$ positive-weighted (possibly overlapping) cliques such that each edge participates in a set of cliques whose weights sum to its own.
The cliques containing edge $(u,v)$ represent the modules in which genes $u$ and $v$ co-express, and clique weights correspond to the module's strength of effect on co-expression. Note that one can obtain a trivial decomposition by assigning every edge to its own 2-clique with matching weight, but this does not lead to any useful insights about the system. Hence, previous work has relied on the principle of parsimony and aimed to find the smallest number of cliques into which the graph is decomposable.\looseness-1

While \EWCD is NP-Hard, Cooley et al.~\cite{cooley2021} recently showed that it admits a \emph{kernel}\footnote{the exact type of kernelization in~\cite{cooley2021} is called a \emph{compression}: problem instances are reduced to equivalent instances of a closely-related problem} of size $4^k$. A kernelization algorithm is a polynomial-time routine that produces a (smaller) \emph{equivalent} instance -- i.e. the kernelized instance is a YES-instance iff the given instance is a YES-instance.  The kernelization technique of Cooley et al. reduces an arbitrary size instance of \EWCD to an equivalent instance of (an annotated version of \EWCD) with at most $4^k$ vertices. Cooley et al.\ also describe an algorithm for obtaining a valid decomposition of a kernelized instance (if one exists). In practice, their kernelization and decomposition algorithms (together called \cricca\footnote{They also give an integer partitioning-based decomposition algorithm for the restricted case of integral weights.}) are able to solve \EWCD for graphs with $k \leq 11$ cliques in less than an hour. However, many co-expression networks have dozens of modules. Thus, a natural question is:

\textit{Does there exist a smaller kernel and/or a faster decomposition algorithm for the \EWCDfull problem?}

We answer this question in the affirmative, giving a $k2^k$-kernel and a faster decomposition algorithm (together called \decaf) which in practice give at least two orders of magnitude reduction in running time over \cricca\footnote{To be precise, over \criccas, an optimized version of \cricca}. Our kernelization technique uses a generalization of Fisher's inequality (from combinatorial design theory). We implement our algorithms and empirically evaluate them on an expansion of the corpus used in~\cite{cooley2021}, demonstrating significantly improved practicality. We first summarize our contributions, then briefly describe the key ideas behind our approaches in~\cref{sec:ideas}.




\textbf{Smaller Kernel:} We give a new kernel reduction rule (\KR 2+) that leads to a kernel of size at most $k2^k$. This is an exponential reduction in the size of the kernel compared to that of \cricca which gives a kernel of size $4^k$ in the worst case. In practice, the kernel obtained using \decaf is at least an order of magnitude smaller than the kernel obtained using \cricca (Figure~\ref{fig:exp-kernel}), which then reduces runtime for the downstream decomposition algorithm.


\textbf{New Search Rules:} The decomposition approach for \EWCD given in ~\cite{cooley2021} is a backtracking algorithm that iteratively searches the space of potential solutions. Every time the algorithm builds a partial solution, it invokes an LP-solver to determine the weights of the cliques. Based on our insights about the structure of such solutions, we are able to design several new \SRsfull (\SRs) that prune away large parts of the search tree reducing the number of times the LP-solver needs to be called. In conjunction with the smaller kernel, these lead to upto three orders of magnitude decrease in the number of runs of the LP-solver (Figure~\ref{fig:lpruns}).

\textbf{Faster Solution to \EWCD:} Figure~\ref{fig:runtime} shows the ratio of the total time taken by \decaf and the total time taken by \cricca for solving \EWCD for several graphs with varying ground-truth $k$. Because of the smaller kernel and new \SRs, we are able to obtain upto two orders of magnitude reduction in the total running time, and the amount of reduction increases as $k$ grows larger.

\textbf{Scale to Larger $k$:} \decaf enables decomposition\footnote{subject to a 3600s timeout (matching that in ~\cite{cooley2021}).} of graphs with $k$ over $50\%$ larger than \cricca (see Figure~\ref{fig:exp-yesno}). 

\begin{figure}
    \centering
    \includegraphics[width=0.45\textwidth]{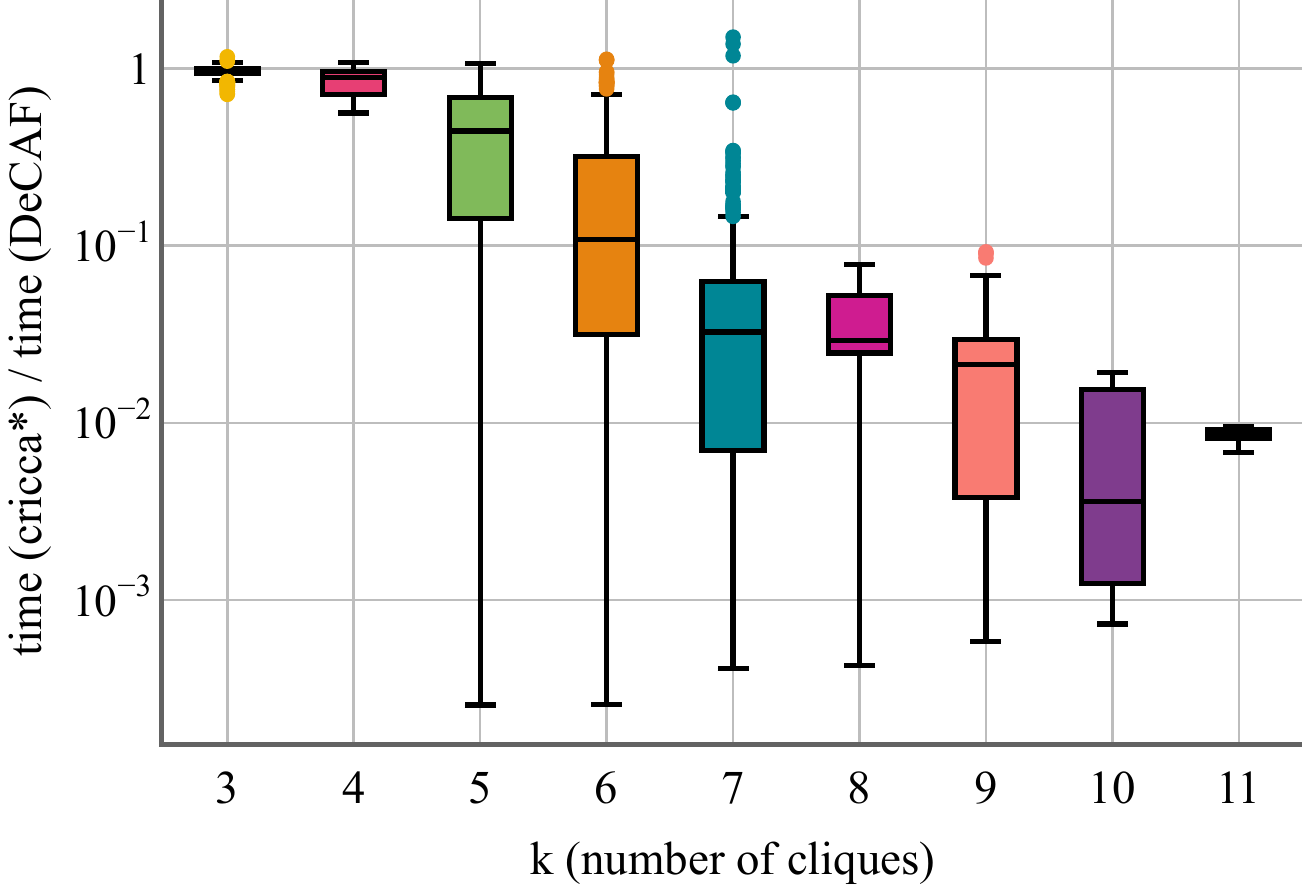}%
    \caption{ \decaf improves the runtime of \criccas by two orders of magnitude. We plot the ratio of execution time on all instances that completed in 3600s with both algorithms. Instances are binned by $k$, the number of cliques in the desired decomposition, highlighting that the improvement in runtime increases with the parameter value.}%
\label{fig:runtime}%
\end{figure}

\section{Preliminaries} \label{sec:prelimineries}

\notes{
\begin{enumerate}
\item Annotated and matrix formulation of the problem.
\item Explain *-equivalence
\item Explain what blocks are
\item Question: Are we going to use the term ``signature''? Madi's paper calls this a row vector.
\item State Fisher's inequality and informally describe what it means. (may move this to the section on kernel).
\item Size: 1 column. 2.75 pages end here.
\end{enumerate}
}

We begin by formally defining \EWCD, first introduced by Cooley et al. as a combinatorial model for discovering modules in gene co-expression graphs~\cite{cooley2021}.

\begin{problem}{\EWCDfull}
\Input & a graph $G = (V,E)$, a non-negative weight function $w_e$ for $e\in E$, and a positive integer $k$.\\
\Prob & a set of at most $k$ cliques $C_1, \ldots C_k$ with weights $\cw_1, \ldots \cw_k\in \R^+$ such that
$w_{uv} = \sum_{i:{uv} \in C_i} \cw_i$ for all $uv \in E$ (if one exists, otherwise output NO).\\
\end{problem}

To solve this problem, the authors of~\cite{cooley2021} also introduced a generalization called \AEWCDfull (\AEWCD) that allowed some vertices to have positive weights.

\begin{problem}{\AEWCDfull}
\Input & a graph $G = (V,E)$, a non-negative weight function $w_e$ for $e\in E$, a special set of vertices $S\subseteq V$, a non-negative weight function $w_v$ for $v\in S$, and a positive integer $k$.\\
\Prob & a set of at most $k$ cliques $C_1, \ldots C_k$ with weights $\cw_1, \ldots \cw_k\in \R^+$ such that
$w_{uv} = \sum_{i:{uv} \in C_i} \cw_i$ for all $uv \in E$ and
$w_{v} = \sum_{i:{v} \in C_i} \cw_i$
for all $v \in S$
(if one exists, otherwise output NO).\\
\end{problem}

Thus, \EWCD is the special case of \AEWCD when the set $S=\varnothing$. The authors showed that an instance of \EWCD can be reduced to an equivalent instance of \AEWCD with at most $4^k$ vertices~\cite{cooley2021}.
Their reduction techniques are based on those of Feldman et al.~\cite{feldmann2020}, who consider a closely related problem \WECPfull (\WECP) and its generalization \AWECPfull (\AWECP). \WECP is the special case of \EWCD when the clique weights are restricted to be $1$. Specifically, for a graph $G$ with edge weights $w_e$ and $k \in \mathbb{Z}^+$, $(G,k)$ is a YES-instance of \WECP iff there is a multiset of at most $k$ cliques such that each edge appears in exactly $w_e$ cliques.




Both these works heavily use linear algebraic techniques, and give {equivalent} matrix problem formulations in which matrices are allowed to have wildcard entries denoted by $\star$.
For $a,b\in \R\cup \left\{ \star \right\}$, let $a\es b$ if either $a=b$ or $a=\star$ or $b=\star$.
For matrices $M1$ and $M2$, we write $M1\es M2$ if $M1_{ij}\es M2_{ij}$ for each $i,j$.
The reformulation of \AEWCD (given by~\cite{cooley2021}) is called \WBSDDWfull (\WBSDDW).

\begin{problem}{\WBSDDWprob}
\Input &
a symmetric matrix $A\in \left(\R^+_0\cup\{\star\}\right)^{n\times n}$ with wildcards appearing on a subset of diagonal entries, and a positive integer $k$\\
\Prob &
a matrix $B\in \bin^{n\times k} $ and a diagonal matrix $W\in (\R_0^+)^{k\times k}$ such that
${A\es BWB^T}$ (if such (B,W) exist, otherwise output NO).\\
\end{problem}

Essentially, the goal is to find an $n \times k$ binary matrix $B$ and a $k \times k$ diagonal matrix with non-negative elements, $W$ such that ${A\es BWB^T}$. The matrix $A$ represents the weighted adjacency matrix of $G$ i.e. $A_{i,j}$ represents the weight of the edge $(i,j)$. If edge $(i,j)$ does not exist in $G$, then $A_{i,j}=0$. In the diagonal matrix $W$, each column (and row) represents a clique from the solution and the element $W_{i,i}$ represents the weight of the $i^{th}$ clique. The wildcard entries in $A$ are used for vertices with no weight restrictions. \EWCD is thus the special case of \WBSDDW where all the diagonal entries are wildcards.

We say that two distinct vertices $u$ and $v$ in G are $\star$-twins if they are
adjacent and satisfy $A_u \es A_v$ . We partition the vertices of $G$ (and correspondingly the rows of $A$) into sets of vertices called \textit{blocks} such that vertices $u$ and $v$ belong to the same block iff $u$ and $v$ are $\star$-twins. ~\cite{feldmann2020} showed that blocks are essentially equivalence classes.

We use $M_i$ to represent the $i^{th}$ row of a matrix $M$. 
The row vector $B_i$ represents the membership information of the $i^{th}$ vertex, i.e. $B_{i,j}$ represents whether the vertex $i$ is a member of the $j^{th}$ clique. We call $B_i$ the \textit{signature} of the vertex $i$. 

\subsection{Related Work} \label{sec:related}
\notes{
\begin{enumerate}
\item Prior work in clique decomposition. Heavy overlap with Madi's paper though not as much emphasis on the time complexity.
\item Maybe some work in biology that uses decomposition into structures other than cliques (if such work exists)
\item Size: about 0.5 to 0.75 columns
\end{enumerate}
}

There is a rich history of work emphasizing the importance of mining patterns in gene co-expression analysis ~\cite{xiao2014multi,yan2007graph,Visscher2017-zd,Nelson2015-ag} in which information about pairwise correlations for all genes in organisms~\cite{Mercatelli2020} is used to derive useful knowledge about sets of genes whose expression is consistently modulated across the same tissues or cell types. Typically, unsupervised network-based learning approaches~\cite{Mao2019,Taroni2019,Leeuw2015} are used for mining such modules. The first combinatorial
model of module identification was \EWCDfull (\EWCD), introduced in~\cite{cooley2021}. They gave a $4^k$-kernel and two parameterized algorithms for obtaining the decomposition of a graph into weighted cliques, one based on linear programming (LP) that works for real-valued weights, and an integer partitioning-based algorithm for the restricted case of integral weights. Both algorithms performed comparably so we use the unrestricted LP-based algorithm.

If clique weights are restricted to all being 1, \EWCD is equivalent to \WECPfull (\WECP) as studied by ~\cite{feldmann2020}. Their work builds upon the linear algebraic techniques of Chandran et al.~\cite{Chandran2016Biclique} for solving the \probsty{Biclique Partition} problem. \WECP itself generalizes \probsty{Edge Clique Partition} (ECP)~\cite{ma1988complexity} which seeks a set of cliques containing each edge at least once (but no constraint on the maximum number of occurrences). It is known that ECP admits a $k^2$-kernel in polynomial time~\cite{MujuniRosamond2008Kernel}.

On the matrix factorization side, where~\cite{cooley2021} modeled \EWCD as the \WBSDDW problem, several other  minimization objectives have been studied.  Zhang et al.~\cite{zhang2013symmetricBMF} minimized ${\|A-BB^T\|_2^2}$ and Chen et al.~\cite{chen2021Instahide} considered ${\|A-BB^T\|_0}$. However, neither of these formulations allow wildcard entries, and hence they are unable to model the clique decomposition problem. The \probsty{Off-Diagonal Symmetric Non-negative Matrix Factorization} Problem studied by Moutier et al.~\cite{offdiagonal} allows diagonal wildcards but also allows $B$ to be any non-negative matrix (not just binary).

\section{Main ideas} \label{sec:ideas}
\notes{
\textbf{Kernel}
\begin{enumerate}
\item Describe briefly what kernelization means.
\item Describe informally the two reduction rules from Feldman's paper
\item Mention that Madi's paper shows that these RRs apply also to their setup.
\item Mention the running time.
\item Dig deeper into why the RRs from Feldman's paper worked: if the graph is a YES instance, then there exists a solution in which all vertices in blocks with $> 2^k$ vertices must have the same signature. This means that the weight of edge $(u,v)$ where $u,v$ are in the block must equal the number of cliques that $u$ and $v$ are involved in in the solution. The weight we set for vertex $v$ when we retain $v$ and remove the rest ensures this important condition is enforced.
\item Such blocks are called blocks of identical twins and others are called fraternal.
\item Mention that fraternal twins have to have different signatures, briefly describe why.
\item We observe that the above reduction can be applied to any block of identical twins. So the main question is: how to distinguish if a block is identical or fraternal?
\item Our central claim is that if a graph is a YES instance then any block with $>k$ vertices must be a block of identical twins.
\item We essentially show that if a block were to be of size $\geq k$ and was a block of fraternal twins, then we could need more than $\geq k$ cliques in the decomposition and we use a celebrated result called Fisher's inequality to prove this.
\item Give final kernel size.
\item Mention that in practice we observe that we are able to get a drastic reduction in instance size.
\end{enumerate}

\textbf{Performance tuning rules}
\begin{enumerate}
\item Describe the working of the decomposition algorithm briefly: Recursive(?) backtracking algorithm. Tries all possible ways of assigning signatures to vertices one by one, backtracking in case of an incompatibe assignment, until either it runs out of signatures to try (NO instance) or finds a valid assignment (YES instance).
\item Finding rules that can detect early on that an assignment cannot lead to a solution can prevent us from exploring large parts of the recursion tree.
\item We introduce 2 new performance tuning (PT, we should find some other name for these though) that help rule out incompatible assignments.
\item Describe first rule. Fairly straightforward
\item Describe second rule (assuming we have renamed rule 3 as rule 2. Please see comment about this on Slack)
\item Question: what do we want to say about running time? Does it go up by a factor of O(n)? Depends on implementation.
\item Mention that in practice these rules significantly speed-up the algorithm.
\end{enumerate}

2.25 columns. Page 5 ends here
}

The starting point for this work is the algorithm of ~\cite{cooley2021} for \EWCD. The algorithm first preprocesses the graph to remove isolated vertices from $G$ and adjusts $k$ accordingly (each isolated vertex must be a unique clique in every valid decomposition). It then uses the kernelization techniques of ~\cite{feldmann2020} to obtain a smaller equivalent instance of \AEWCD. On this kernelized instance, it runs a (parameterized) clique decomposition algorithm that searches the space of clique membership signatures for each vertex. Note that the kernel reduction rules of ~\cite{feldmann2020} were proposed for the \AWECP problem in which clique weights are restricted to be $1$. ~\cite{cooley2021} showed that the same rules can be used to obtain a kernel for \AEWCD. We will follow a similar sequence by first giving a kernel reduction rule for \AWECP and then showing that it can also be used for \AEWCD. Thus, we first consider the setup of the \AWECP problem. 

The kernelization technique of ~\cite{feldmann2020} works by applying two reduction rules to blocks of $\star$-twins which either prune away many of the vertices from the block or act as easy checks for a NO-instance.

\begin{krr}[\cite{feldmann2020}]
    If there are more than $2^k$ blocks, then output that the instance is a NO-instance.
\end{krr}

\begin{krr} [\textnormal{informal,} \cite{feldmann2020}] If there is a block $D$ of size greater than $2^k$, then pick two distinct $i, j \in D$. Convert $G$ into an instance $G'$ of \AEWCD by setting the weight of vertex $i$ equal to the weight of edge $(i,j)$  and removing every vertex in $D$ from $G$ except $i$. Then $(G,k)$ is a YES-instance iff $(G', k)$ is a YES-instance.
\end{krr}


For \KRfull 1 (\KR 1), the authors ~\cite{feldmann2020} show that if $G$ is a YES-instance, then any pair of vertices $u,v$ in $G$ such that $u$ and $v$ have different signatures in the solution must belong to different blocks. Since there can be at most $2^k$ possible signatures (binary vectors of length $k$) there can be at most $2^k$ blocks in a YES-instance.

For \KRfull 2 (\KR 2), they first prove that if $G$ is a YES-instance then for any block of vertices, there exists a solution in which all the vertices from the block either have the same signature or all have pairwise distinct signatures. Since there can be at most $2^k$ distinct signatures, if the number of vertices in a block is greater than $2^k$ then by the pigeonhole principle, the signatures cannot all be distinct. Hence, there must exist a solution in which all the vertices in the block have the same signature. In such a case, we can keep just one representative vertex from the block to get a smaller instance.

Note that this means that any two vertices of such a block must participate in exactly the same set of cliques. For this to happen, for vertices $u,v$ in the block, the weight $w$ of the edge $(u,v)$ must equal the number of cliques that contain $u$ and $v$ in the solution. Any solution for the reduced instance that is extendable to a solution for the original instance must ensure that the representative vertex $v$ is part of exactly $w$ cliques. To enforce this condition the weight of the representative vertex is set to $w$ in the reduced instance.

Thus, after applying \KRs 1 and 2 there are at most $2^k$ blocks with at most $2^k$ vertices in each block. In this way, the authors obtain a kernel of size $4^k$. Although this kernel was proposed by ~\cite{feldmann2020} for \WECP in which the clique weights have to all be 1, ~\cite{cooley2021} showed that the same techniques apply even when the cliques are allowed non-unit weights.

\textit{Our main insight is that \KR 2 applies to a broader set of blocks.} Specifically, we show that if $G$ is a YES-instance then for any block with more than $k$ vertices, all vertices of the block must have identical signatures. For this, we use Fisher's inequality to prove that if a block has $\geq k+1$ vertices and not all vertices of the block participate in the same cliques, then the number of cliques required to cover the edges in the block is at least $k+1$, contradicting the assumption that the given instance is  a YES-instance. Thus, we apply \KR 2 to every block with $\geq k+1$ vertices. Since there are at most $2^k$ blocks and each block has at most $k$ vertices, we obtain a $k2^k$-kernel. Although this gives a kernel for the \WECP problem (in which cliques are constrained to have unit weight), similar to ~\cite{cooley2021} we show that this works for \EWCD as well. The running time of the kernelization algorithm remains unchanged at $O(n^2\log n)$. In practice, we observe that this \KR gives an order of magnitude reduction in the size of the kernel, which subsequently helps speed up the decomposition.

\textbf{\SRsfull (\SRs):} The decomposition algorithm of ~\cite{cooley2021} uses backtracking to assign signatures to vertices one by one. While doing so, it checks to make sure that the new signature is compatible with the vertices that have already been assigned a signature. If no compatible signature is found for a vertex, the algorithm backtracks. It continues this process until all vertices have a valid assignment or it determines that no valid assignment can be found. Having rules that can quickly detect that the current (partial) assignment cannot be extended to a valid decomposition helps us prune away branches of the search tree and reduce the running time.

Our first \SR pertains the order in which we consider the vertices for signature assignment, which
can significantly affect how much of the search tree must be explored. We tested several strategies, and found that a \texttt{push\_front} approach in which vertices from reduced blocks are considered before those in non-reduced blocks was most effective. A justification for this strategy is in~\cref{sec:decomposition}, and the empirical evaluation is shown in \cref{subsec:exp_order}.

Our second \SR comes from the straightforward observation that when assigning a signature to a vertex, one must respect its non-neighbor relationships. More specifically, when finding a signature for a vertex $v$, for all $u$ such that $(u,v)$ is not an edge, $B_uB_v^T = 0$. That is, if $u$ and $v$ are non-neighbors, they cannot share a clique. Thus, we only test those signatures for $v$ that have no cliques in common with the signatures of its non-neighbors.

 Our third and final \SR makes use of the fact that if the graph is a YES-instance then there exists a solution in which all vertices in each block have either identical signatures or pairwise distinct signatures. Thus, when assigning signatures to the vertices of a block, either we assign a unique signature to all vertices in the block or assign the same signature to all vertices in the block. This eliminates assignments in which a block has the same signature appearing on more than $1$ but but not all vertices in the block.

The improved kernel along with these \SRs speeds up decomposition by two orders of magnitude. We give formal proofs of our results in the next two sections.

\section{Smaller kernel} \label{sec:kernel}

Similar to ~\cite{cooley2021}, we will first consider the setup of the \AWECP problem (in which the cliques are constrained to have weight 1), and prove correctness of our new reduction rule.
We will then show that the rule remains valid in the case of \AEWCD. We begin by restating an important lemma from~\cite{feldmann2020}.

\begin{lemma}[Restated Lemma 8 from ~\cite{feldmann2020}] \label{lem:submatrix}
For a block $D$, the entries of the sub-matrix $A_{D,D}$ are all same except for wildcards.
\end{lemma}

An important property of $\star$-twins implied by \cref{lem:submatrix} is: if vertices $u$ and $v$ are $\star$-twins and $u$ is annotated, then $A_{uu} = A_{uv}$. In other words, the weight of $u$ must equal the weight of the edge $(u,v)$ because otherwise $A_{uv}$ is not $\es$ to $A_{vv}$, contradicting the fact that $u$ and $v$ are $\star$-twins. Any block (definitionally) consists exclusively of $\star$-twins, thus, the weight of every edge within the block must be the same. In other words, the vertices in the block form a complete subgraph with uniform edge-weights.

An important observation about any solution for \WECP (\EWCD) is that since the vertices do not have weights, if the solution contains a singleton clique (i.e. a clique of size 1) and if we remove this clique the remaining set of cliques also forms a solution. (Recall that isolated vertices have already been removed in the preprocessing step). Thus, \textbf{in the rest of this paper, we restrict our attention to solutions that do not contain singleton cliques}. Note that, on the other hand an instance of \AWECP (\AEWCD) obtained from an instance of \WECP (\EWCD) \textit{can} have singleton cliques to satisfy vertex weights. However, such vertices must necessarily be the representative vertices of reduced blocks.



We now formally define two types of $\star$-twins.

\begin{definition}[\textbf{Identical and fraternal twins}]
Given a YES-instance $(G,k)$ of \AWECP, $\star$-twins $u$ and $v$ in $G$ are called \textit{identical twins} if there exists \textit{no} solution in which $u$ and $v$ have distinct signatures, and \textit{fraternal twins} otherwise.
\end{definition}





Note that there can be solutions in which fraternal twins have identical signatures -- we only require that in \emph{some} solution they have different ones. Moreover, since a block consists of $\star$-twins, if any two vertices in a block are identical twins, then all must be.
We call such blocks \emph{identical blocks}.

\KR 2 from ~\cite{feldmann2020} implies that all blocks with $>2^k$ vertices are identical blocks. Our main insight is that a broader set of blocks must have this property. More specifically, any block with $>k$ vertices must be an identical block.

To prove this, we will use a result from combinatorial design theory known as the non-uniform Fisher's inequality:

\begin{theorem}[\textnormal{restated from }~\cite{mathew2020combinatorial}]\label{thm:fisher}
Let $w$ be a positive integer and let $A = \{A_1,...,A_t\}$ be a family of subsets of $U = \{e_1,...,e_r\}$. If $|A_i \cap A_j| = w$ for each $1 \leq i < j \leq t$, then $|A| =t \leq r = |U|$.
\end{theorem}

Essentially, the non-uniform Fisher's inequality states that if we have a set of $r$ elements and we form $t$ subsets of these elements such that any two subsets intersect in exactly $w$ elements, then the number of subsets can be at most the number of elements.

Fisher's inequality was first proposed in the context of Balanced Incomplete Block Design (BIBD) (See ~\cite{stinson2004introduction} and ~\cite{babai1988linear}). The uniform version was first proposed by Ronald Fisher 
and Majumdar~\cite{majumdar1953some} showed that the inequality holds even in the non-uniform case. The inequality has since been proven and applied in many different problem areas. In fact, De Caen and Gregory ~\cite{de1985partitions} showed the following corollary which, as we will show below, directly corresponds to the problem we're considering.

Let $K_t$ represent the unweighted, complete graph on $t$ vertices, and $wK_t$ denote the complete multigraph on $t$ vertices where the multiplicity of every edge is $w$. Let us say that a partition $R$ of the edge-set of $wK_t$ into cliques is \emph{non-trivial} if there exists any clique in $R$ that is not $K_t$. \cref{cor:decaen} implies if $R$ is non-trivial then $t \leq |R|$. This can also be viewed as a generalization of the clique partition theorem of De Bruijn and Erd\"os ~\cite{de1948combinatorial} for arbitrary $w$ (the result of ~\cite{de1948combinatorial} was for $w=1$).


\begin{corollary}[restated Corollary 1.4 from ~\cite{de1985partitions}] \label{cor:decaen}
Let $R$ be a partition of the edge-set of $wK_t$ into non-empty cliques. If not all cliques in $R$ are $K_t$ then $|R| \geq t$.
\end{corollary}

We are now ready to prove our main theorem.
Let $D$ be a block with $t=|D| \geq k+1$ vertices. Let $w$ be the weight of the edges in $D$.

\begin{theorem}\label{thm:main}
If the given instance $(G,k)$ is a YES-instance of \AWECP, then $D$ must be an identical block.
\end{theorem}

\begin{proof}

Suppose $D$ is not identical. Since $G$ is a YES-instance, there exists a solution such that not all vertices from $D$ have the same signature i.e. not all vertices from $D$ appear in the same cliques. Let $S$ be the multiset of cliques in such a solution. Let $S_D \subseteq S$ be the multiset of cliques in which vertices from $D$ appear. From every clique $C \in S_D$, delete all vertices that are not in $D$ and call the resultant multiset $R$. Thus, $R$ is a multiset of subsets of $D$. Since any subset of vertices in a clique also form a (smaller) complete graph, $R$ is a multiset of cliques. One can think of the cliques in $R$ as the ``projection'' of the cliques in $S_D$ onto $D$. Since, not all vertices from $D$ appear in the same cliques, $R$ must consist of a non-trivial clique.

We first show that the non-trivial cliques in $R$ cannot all be singleton cliques. 
Suppose there is a singleton clique in $R$ consisting of $v \in D$. Then there exists a non-singleton clique $C$ in $S_D$ such that $C \cap D = v$. Moreover, since $C$ is not singleton, there exists a vertex $v' \in C \setminus D$. Let $u \neq v$ be a vertex in $D$. Since $u$ and $v$ are $\star$-twins and $v$ is a neighbor of $v'$, $u$ must also be a neighbor of $v'$ and $w_{uv'}=w_{vv'}$. Thus, there must exist a clique $C' \in S_D$, $u,v' \in C', v \notin C'$. In other words, every vertex $u \in D, u \neq v$ must be a part of some clique that $v$ is not a part of. If each such clique projects into a singleton clique in $R$, then $|R| > |D|=k+1$ which is a contradiction since $G$ is a YES-instance. Thus, there must exist a non-singleton, non-trivial clique in $R$ i.e. there must exist a clique containing $>1$ but not all vertices from $D$. 

  Let $H$ be an unweighted, complete multigraph on $|D|$ vertices in which the multiplicity of every edge is $w$. It is easy to see that $R$ is a partition of the edge-set of the multigraph $H$ into cliques. If $R$ consists of a non-singleton, non-trivial clique then by ~\cref{cor:decaen}, $|D| \leq |R|$, which is a contradiction because $G$ is a YES-instance (implying $|R| \leq k$ and $|D| \geq k+1$). Thus, $R$ cannot consist of non-trivial cliques.

In other words, $R$ must be trivial i.e. every clique in $R$ must be a $K_t$. Thus, every vertex of $D$ must be in the exact same set of cliques in $S$ and have the same signature in the solution matrix $B$. Thus, $D$ must be an identical block.


\end{proof}

Hence, we can reduce it to a representative vertex, leading to our enhancement of \KR 2:


\begin{customkrr}{2+}\label{rr:enhanced}
    If there is a block $D$ of size greater than $k$, then apply the reduction of \KR 2.
\end{customkrr}





We further show that \KR 2+ gives a valid kernel even in the case of \AEWCD. Our proof of correctness closely follows that of rule 2 in ~\cite{cooley2021} and proceeds in two parts. Proofs of these Lemmas are deferred to Appendix~\ref{app:proofs}.

\begin{lemma} \label{lem:correctness-rr2-yes}
    Let $(A',k)$ be the reduced instance constructed by applying \KR 2+ to $(A,k)$. Then if $(B',W')$ is a solution for $(A',k)$  then the $(B,W)$ constructed by \cref{rr:enhanced} is indeed a solution to $(A,k)$.
\end{lemma}

\begin{lemma}
    \label{lem:correctness-rr2-no}
    If $(A,k)$ is a YES-instance, then
    the reduced instance $(A',k)$ produced by \cref{rr:enhanced} is a YES-instance.
\end{lemma}

Thus, our kernelization algorithm applies \KR 1 of ~\cite{cooley2021} and \KR 2+ to get a kernel that is at most $k2^k$ in size. The kernelization algorithm sorts the rows in $A$ to form blocks and reduces each block that has $\geq k+1$ vertices. The sorting of rows of $A$ and division into blocks takes time $O(n^3)$ and the application of \KR 1 and \KR 2+ take time $O(n)$. Thus, the time complexity of the kernelization algorithm is $O(n^3)$, matching that of ~\cite{cooley2021}.

\notes{
\begin{enumerate}
\item Define identical twins
\item Define fraternal twins
\item Prove that every vertex in a block of fraternal twins has to have unique signature.
\item Prove our main claim using Fisher's inequality
\item State the final theorem about kernel size.
\item 2 columns. Page 6 ends here.
\end{enumerate}
}

\section{Faster Decomposition Algorithm} \label{sec:decomposition}

\notes{
\begin{enumerate}
\item We should probably put the pseudocode for this in the Appendix. Do you agree?
\item Briefly describe PT1.
\item Prove correctness of PT1
\item Prove that every block of identical twins has a unique signature.
\item Describe PT2. Mention that it follows from the definitions of identical and fraternal twins.
\item Question:Time complexity bound? In terms of what?
\item Size: 1.75 columns. Page 6.75 ends here.
\end{enumerate}
}


Once a kernelized instance is obtained, one can run the decomposition algorithm of ~\cite{cooley2021}, shown in \AlgLp (\cref{alg:Lpmain}), on it. The algorithm assigns signatures to vertices one-by-one, iteratively building a solution.  When only a subset of all vertices have been assigned a signature, we call this a \textit{partial assignment}. When trying to find a compatible signature for a vertex, the algorithm searches the entire space of $2^k$ possible signatures for that vertex. For every signature that the algorithm considers for a vertex, the clique weights as given by the weight matrix $W$ may need to change. \Comp (\cref{alg:comp}) checks if the weight matrix $W$ is compatible with the current partial assignment. If not, to find a compatible new set of weights, the algorithm builds a Linear Program (LP) which encodes the partial assignment and the edge-weights (\InferCWLp, \cref{alg:InferCWLp}). If the LP returns a solution, the algorithm updates the weight matrix. If the LP fails to return a solution, it means that no feasible weight matrix exists for this partial assignment. In this case, the algorithm backtracks. Since \InferCWLp and \Comp are not affected by our search reduction rules, we defer their pseudocode to the appendix.

\begin{algorithm}[h]
\begin{algorithmic}[1]
    \For{$P \in \{0, 1 \}^{2k \times k}$}\label{line:LpFor}
\State{initialize $\basis$ to a $n \times k$ null matrix}
\State{$b, i \leftarrow 1$}
\While{$b \leq 2k $ }\label{line:LpWhile}
\State{$\basis_i \leftarrow P_b$}\label{line:LpBasisFill}
        \State{$b\leftarrow b+1$}
        \State{$W\leftarrow$ \InferCWLp($A$, $\widetilde{B}$)}
        \If{$W$ is not null matrix}
            \State{$(B, i) \leftarrow$ \FillNB(A, $\widetilde{B}$, $W$)}\label{line:LpFillNBCall}
            \If{$i=n+1$}
                \Return {($B$, $W$) }\label{line:LpYesReturn}
            \EndIf
        \Else
            $\;b \leftarrow 2k+1$ \Comment{\small{\it null $W$; break out of while}}
        \EndIf
    \EndWhile
\EndFor
\State \Return No \label{line:LpNoReturn}
\end{algorithmic}
\caption[]{\hspace*{-4.3pt}{.} \AlgLp}\label{algLppsuedo}
\label{alg:Lpmain}
\end{algorithm}

The main insight of ~\cite{cooley2021} was that the pseudo-rank of $B$ is at most $2k$ and that once the basis vectors of $B$ have been guessed correctly, there will be no need to backtrack when filling in the signatures of other vertices. They showed that we need to run the LP only when the algorithm backtracks and that this happens for at most $2^{2k^2}$ partial assignments. \FillNB (\cref{alg:fillNB}) shows the pseudocode for filling in the signatures for the non-basis vectors. Note that every time the algorithm picks a new set of basis vectors ($P$), the existing assignment of signatures (even non-basis vectors) are discarded. After an exhaustive search if no valid assignment is found, the algorithm returns that the instance is a NO-instance.

\begin{algorithm}[!h]
\begin{algorithmic}[1]
\State $B\leftarrow\basis$
\While{$B$ has a null row}\label{line:fillnbwhile}
    \State \textbf{let} $B_i$ be the first null row
    \For{$v\in \left\{ 0,1 \right\}^k$}\label{line:fillnbfor}
    \If{\Comp($A$, $B$, $W$, $i$, $v$)}\label{line:compcall}
    \State $B_i\leftarrow v$ \label{line:BFill}
        \State \textbf{goto} line 2
    \EndIf
    \EndFor
    \State \Return $(B,i)$ \Comment{\small{\it there is no $(i,W)$-compatible $v$}}
\EndWhile
\State \Return $(B, n+1)$ \Comment{\small{\it B has no null row}}
\end{algorithmic}
\caption[]{\hspace*{-4.3pt}{.} \FillNB($A$, $\widetilde{B}$, $W$)}\label{FillNonBasis}
\label{alg:fillNB}
\end{algorithm}


We now design several search reduction rules that can help to quickly prune away branches that cannot lead to a solution.




\textbf{\SRfull 0}: \textit{Assign signatures to reduced blocks before non-reduced blocks.} We observed during experiments that the order in which we assign signatures to vertices can significantly impact how fast the algorithm terminates. Ideally, we would like the first $k$ vertices to be as close to an independent set as possible (since non-neighbors significantly restrict potential valid signatures, see \SR 1). Since vertices within a block must be neighbors, we wanted an order that hit many distinct blocks quickly, yet was compatible with the engineering required for \SR 2 (below). This led to the strategy \texttt{push\_front}, which assigns signatures to all vertices which are representatives of reduced blocks (which necessarily have size 1) before proceeding to those in non-reduced blocks. We validated our choice by empirically evaluating this against several other orders including the \texttt{arbitrary} approach in ~\cite{cooley2021}; details are in ~\cref{subsec:exp_order}.



\textbf{\SRfull 1:} \textit{For every vertex, generate only those signatures that don't share a clique with the non-neighbors of that vertex.} The main idea behind \SR 1 is that any two non-neighbors should not share a clique. Thus, for non-neighbors $u$ and $v$, $B_u{B_v}^T=0$. Whenever the algorithm is searching for a signature to assign to a new vertex, it generates the list of cliques that are ``forbidden'' for that vertex based on the signatures of the vertex's non-neighbors. It then uses this list to generate only those signatures that respect the ``forbidden'' cliques. In many cases, this drastically reduces the number of signatures to be tried.

\textbf{\SRfull 2}: \textit{Vertices across blocks must have unique signatures. The cliques that a vertex participates in cannot be a proper subset of the cliques its $\star$-twin participates in. Make all signatures in a block either identical or pairwise distinct.}

We first prove that signatures cannot be shared across blocks. 

\begin{lemma}
If $u$ and $v$ belong to different blocks then $B_u \neq B_v$.
\end{lemma}
\begin{proof}
For contradiction, assume $B_u=B_v$. Since $u$ and $v$ belong to different blocks, they must not be $\star$-twins. Thus, either $u$ is not adjacent to $v$ or $A_u \stackrel{\star}{\neq} A_v$. 

If $u$ is not adjacent to $v$, $A_{uv}=B_uWB_v^T=0$. Since, $B_u=B_v$ and since we do not allow cliques to have weight $0$, $B_uWB_v^T=0$ iff $B_u=B_v=0$. Thus, it must be the case that $B_u=B_v=0$. However, since the graph has already been preprocessed to remove all isolated vertices, the remaining vertices must all have at least one edge adjacent to them and hence must be a part of at least one clique. Thus, $B_u, B_v \neq 0$ which is a contradiction.


Now consider the case when $u$ and $v$ are adjacent but $A_u \stackrel{\star}{\neq} A_v$. Since $B_u=B_v$, $B_uWB^T = B_vWB^T$. Since $B_uWB^T \es A_u$ and $B_vWB^T \es A_v$,  this implies that $A_u \es A_v$, a contradiction. 
\end{proof}


We will now show that there cannot be $\star$-twins such that the signature of one is a subset of the signature of the other. Note that when we apply the \SRs, the graph has already been preprocessed. Thus, there are no isolated vertices in $G$.


\begin{lemma}
Consider a YES-instance $(G,k)$ of \AEWCD. There exists a solution $B$ such that for every pair of  vertices $u$ and $v$ that are $\star$-twins in $G$, $B_u \nsubset B_v$ and vice versa.
\end{lemma}


\begin{proof}
Suppose there exist $\star$-twins $u$ and $v$ and a solution $B$ such that $B_u \subset B_v$. Then there exists a clique $C$ in the solution such that $v \in C, u \notin C$. If $C$ is a singleton clique or $C$ has weight $0$, then we can remove $C$ from the solution. Clearly, the remaining set of cliques would still be a solution and the theorem would be true. So assume $C$ has positive weight and is not a singleton clique. Thus, there exists a vertex $v' \in C, v' \neq v$. Since $u$ and $v$ are $\star$-twins and $v$ is a neighbor of $v'$, $u$ must also be a neighbor of $v'$ and $w_{uv'}=w_{vv'}$. Let $S_{uv'}$ and $S_{vv'}$ be the set of cliques in the solution in which edges $(u,v')$ and $(v,v')$ participate, respectively. Then since $B_u \subset B_v$, $S_{uv'} \subset S_{vv'}$. Moreover, $w_{uv'}=\sum_{i:{uv'} \in C_i} \cw_i$ where $\cw_i$ represents the weight of clique $i$. Similarly, $w_{vv'}=\sum_{i:{vv'} \in C_i} \cw_i$. But since the cliques have positive weights and $S_{uv'} \subset S_{vv'}, w_{uv'} < w_{vv'}$ which is a contradiction.
\end{proof}


We know from \cref{thm:main} that blocks having $>k$ vertices must be identical blocks. However, blocks having $\leq k$ vertices can be fraternal. Moreover, by Lemma 11 of ~\cite{feldmann2020} we know that if the given instance is a YES-instance then there exists a solution in which all vertices in a block have either identical signatures or pairwise distinct signatures. We show that if the given instance is a YES-instance of \AEWCD then there exists a solution in which this condition is \textit{simultaneously} true for all blocks.

\begin{theorem}\label{thm:simultaneous}
Given a YES-instance $(G,k)$ of \AEWCD, there exists a solution in which every block of $G$ has either identical signatures or pairwise distinct signatures.
\end{theorem}

\begin{proof}
Consider any solution $(B,W)$ to the given instance and suppose there exists a block $D$ whose vertices have signatures that are neither all unique nor all pairwise-disjoint (if such a block does not exist then the theorem is trivially true). Thus, there exist vertices $v_1, v_2, v_3$ in $D$ such that $B_{v_1} = B_{v_2} \neq B_{v_3}$. Consider the $n \times k$ matrix $B'$ got by setting ${B'}_{v_1}=B_{v_1}$, ${B'}_{v_2}=B_{v_2}$, ${B'}_{v_3}=B_{v_1}$ and $\forall x \neq v_1, v_2, v_3, {B'}_x=B_x$ i.e. $B'$ has the same signatures as $B$ for all vertices except $v_3$ and ${B'}_{v_3}$ is set to $B_{v_1}$. To prove that $(B',W)$ is indeed a valid solution for $(G,k)$, it is sufficient to show that ${B'_u}W{B'}_v^T\es A_{uv}$ for all vertices $u,v$ in $G$.  There are 3 cases to consider:
    \begin{itemize}
    \item Case 1: $u,v\neq v_3$. Then ${B'_u}W{B'}_v^T={B_u}W{B_{v}^T}\es A_{uv}$.
    \item Case 2: $u=v_3$,$v\neq v_3$. Then ${B'_u}WB_v^T={B_{v_1}}WB_v^T\es A_{v_1v}=A_{v_3v}$,
    where the last equality follows from $v_1$ and $v_3$ being in the same block $D$.
    \item Case 3: $u=v=v_3$. In this case, ${{B'}_{v_3}}W{B'}_{v_3}^T={B_{v_1}}WB_{v_3}^T \es A_{v_1v_3}\es A_{v_3v_3}$.
    The $\es$ here follows because any two entries (that are not $\star$) in the same block of matrix $A$ are equal (Lemma 7 of \cite{feldmann2020}).
    \end{itemize}
We can apply this iteratively to all vertices in $D$ and to other blocks until all blocks either have identical signatures or pairwise distinct signatures.
\end{proof}




Thus, we only need to search over those assignments that satisfy the conditions of \cref{thm:simultaneous}. In \cref{line:LpBasisFill} in \AlgLp and \cref{line:fillnbfor} in \FillNB, when assigning a signature to a vertex we ensure that all vertices in its block have either identical or unique signatures. This eliminates assignments in which a block consists of some repeated signature but not all identical signatures, thereby reducing the search space.

\textbf{Running time:} The \texttt{for} loop in \cref{line:LpFor} has at most $2^{{2k}^2}$ iterations. \FillNB takes time $O(n^2k2^k)$ where $n$ is the number of vertices in the kernelized instance. \InferCWLp solves an LP with $k$ variables and at most $4k^2$ constraints which can be solved in $O(k^4L)$ time where $L$ is the number of bits required for input representation~\cite{LpVaidya89}. \Comp runs in time $O(nk)$ and hence, the total time taken by \FillNB is $O(n^2k2^k)$. Hence, the total time taken by \AlgLp without the \SRs is $O(2^{{2k}^2}2k(k^4L + n^2k2^k))$. Since $n \leq 4^k$ the total running time of \AlgLp as given by ~\cite{cooley2021} is $O(4^{k^2}k^2(32^k +k^3L))$. Thus, if $L=O(2^k)$, this comes to $O(4^{k^2}k^232^k)$.

\SR 0 pushes singleton vertices from reduced blocks to front in time $O(n)$. \SR 1 when assigning a signature to a vertex $v$, loops over the non-neighbors of $v$ and generates the list of forbidden cliques in time $O(nk)$. \SR 2 when assigning a signature to a vertex, compares with the signatures of other vertices in the block in $O(nk)$ time. Thus, the running time of  \AlgLp with \SRs is $O(2^{{2k}^2}2nk^2(k^4L + n^2k2^k))$ Due to the new kernel, we can set $n=k2^k$. This gives an overall running time of $O(4^{k^2}2^kk^4(k^22^{3k} +k^3L))$. Thus, if $L=O(2^k)$ this comes to $O(4^{k^2}k^616^k)$. Thus, the running time reduces significantly and in practice we are able to get at least two orders of magnitude speedup.



\section{Experimental Results} \label{sec:results}

We compared the performance of our kernelization and decomposition algorithms with those of ~\cite{cooley2021}. We use the publicly available Python package of ~\cite{cooley2021} provided by its authors and implement our new algorithms as an extension. We ran all experiments on identical hardware, equipped with 40 CPUs (Intel(R) Xeon(R) Gold 6230 CPU @ 2.10GHz) and 191000 MB of memory, and running CentOS Linux release 7.9.2009. We used Gurobi Optimizer 9.1.2 as the LP solver, parallelized using 40 threads. A timeout of 3600 seconds per instance was used in all experiments. Code and data to replicate all experiments are available at~\cite{codebase}, and will be made open-source before publication.

\textbf{Notation:} We use \KD and \SD to denote our kernelization and decomposition algorithms, respectively. We use \KC to denote the kernelization algorithm of ~\cite{cooley2021}, and \SC (\SCS) for the decomposition algorithm of ~\cite{cooley2021} \textit{without (with)} \SR 0.
Thus, \cricca represents the combination (\KC, \SC), \criccas is (\KC, \SCS) and \decaf is(\KD, \SD). We use $n_C$ and $n_D$ to denote number of vertices remaining in the graph (the kernel) after applying \KC and \KD respectively. For some of our experiments, we subselect instances based on their seed. Details can be found in the description of each experiment.

\textbf{Datasets:} The authors of ~\cite{cooley2021} evaluate \cricca on two biology-inspired synthetic datasets (LV and TF, seeds 0-19); we use the same datasets for our experiments. Details about how these datasets were created are in Section 7 and Appendix D of~\cite{cooley2021}. We assume that every graph has been preprocessed to remove isolated vertices. Each (preprocessed) graph $G$ has a ground-truth $k$ -- the minimum $k$ such that $(G,k)$ is a YES-instance. We distinguish this from $k_{in}$, which denotes the number of desired cliques input to the algorithms; each $(G,k_{in})$ denotes a unique instance of \EWCD. Details about these instances can be found in \cref{table:graphstats}. As in~\cite{cooley2021}, we restrict our main corpus to graphs with $3 \leq k \leq 11$; for $k \in [5,7]$, we generated multiple instances for each $G$ using several values of $k_{in}$ above and below $k$ (E.3 in~\cite{cooley2021}). The number of unique $G$ are shown in column ``$\#G$'' and the number of $(G,k_{in})$ instances are shown in the column ``$\#(G,k_{in})$''. The columns $n$ and $m$ denote the average number of vertices and edges across all instances per $k$ value. For ease of comparison of running times, we do not include instances that timed out on any combination of \KC, \KD, \SCS and \SD in \cref{table:graphstats}. A detailed breakdown of timeouts is in the Appendix (\cref{sec:fastresults}).



\begin{table}[h]
\centering
\begin{adjustbox}{max width=0.47\textwidth}
\begin{tabular}{|l|rrrr||p{1.2cm}|p{1.2cm}|p{1.2cm}|p{1.2cm}|}
\toprule
 \multicolumn{5}{|c||}{}    & \multicolumn{2}{c|}{\KC}    & \multicolumn{2}{c|}{\KD}     \\\cline{6-9}
 \multicolumn{5}{|c||}{}    & \multicolumn{1}{c|}{\SCS}    & \multicolumn{1}{c|}{\SD}    & \multicolumn{1}{c|}{\SCS}    & \multicolumn{1}{c|}{\SD}  \\ \cline{6-9}
 \multicolumn{5}{|c||}{}    & \multicolumn{1}{c|}{\criccas}    & \multicolumn{1}{c|}{}    & \multicolumn{1}{c|}{}    & \multicolumn{1}{c|}{\decaf}  \\ 
 \hline
 \multicolumn{1}{|l|}{$k$}  & \multicolumn{1}{r}{$\#G$}  & \multicolumn{1}{r}{$\#(G,k_{in})$}    & \multicolumn{1}{r}{$n$}   & \multicolumn{1}{r||}{$m$} & \multicolumn{4}{c|}{speedup}  \\
\midrule
3 & 321 & 321 & 198.21 & 38.72 & 1.00 & 2.40 & 2.23 & 2.45 \\
4 & 225 & 225 & 223.13 & 146.69 & 1.00 & 3.60 & 4.74 & 8.81 \\
5 & 300 & 564 & 313.85 & 2677.86 & 1.00 & 6.30 & 147.91 & 669.55 \\
6 & 360 & 734 & 369.28 & 10182.26 & 1.00 & 11.83 & 158.88 & 1379.24 \\
7 & 136 & 274 & 542.28 & 28477.22 & 1.00 & 13.06 & 500.44 & 6061.28 \\
8 & 31 & 31 & 459.77 & 22834.32 & 1.00 & 17.79 & 169.95 & 3421.30 \\
9 & 27 & 27 & 824.78 & 61877.89 & 1.00 & 6.45 & 317.75 & 9316.60 \\
10 & 24 & 24 & 778.25 & 116381.50 & 1.00 & 22.96 & 134.39 & 6876.59 \\
11 & 18 & 18 & 953.00 & 99219.67 & 1.00 & 12.66 & 128.11 & 13475.95 \\
\bottomrule
\end{tabular}
\end{adjustbox}
\caption{Corpus number of instances and average sizes grouped by $k$ value (seeds $[0,9]$
and $k \in [3,11]$). Right columns show average runtime reduction for
combinations of kernels \KC, \KD and decomposition algorithms \SCS, \SD.}
\label{table:graphstats}
\end{table}


\subsection{Smaller Kernel}

To evaluate the effect of \KRfull 2+,
we compared the kernel size and the running time of \KD with those of \KC.
We used the entire dataset (seeds 0-19) of ground-truth $k$ between $3$ and $11$ (inclusive),
with input $k$ ($k_\text{in}$) varied between $0.4 k$ and $1.6k$.
Because the difference between two versions is only the threshold value used in
\KR 2 and \KR 2+, it holds that $n_D \leq n_C$.
Further, if \KC \textit{solves} the problem,
that is, it finds the instance a NO-instance, then so does \KD.

Figure~\ref{fig:exp-kernel} plots the ratio $n_D / n_C$ for each instance,
colored by $k_\text{in}$.
We can see that most instances got shrunk to less than 20\% of $n_C$.
Also we observe that our kernel size, $n_D$ decreases (1) as $n$ increases
and (2) as $k_\text{in}$ increases.
The former is intuitive; \KR 2+ is more effective when the instance is larger,
removing more vertices from the graph.
The latter is explained by the fact that \KC is effective only for blocks that are exponential in size with respect to $k_\text{in}$ and thus, has less impact when $k_\text{in}$ is large, where as \KD is able to reduce all blocks with $> k_{in}$ vertices and hence is effective even when $k_{in}$ is large.



For the running time of the kernelization process, we computed the relative running time i.e. the ratio of the running time of \KD to that of \KC.
The first and third quantiles were $0.99$ and $1.01$, respectively,
and the maximum was $1.24$.
We conclude that there is no significant difference in running times.

\begin{figure}[t]
    \centering
    \includegraphics[scale=0.3]{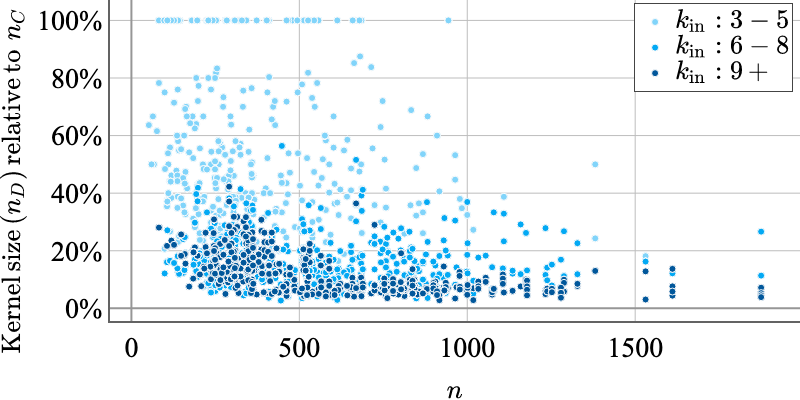}
    \captionof{figure}{%
        Kernel size of $K_D$ relative to that of $K_C$ on entire
        corpus (seeds 0-19), sorted by instance size (prior to kernelization)
        and colored by $k_in$.
    }
    \label{fig:exp-kernel}
\end{figure}

\subsection{Effects of Vertex Reordering}\label{subsec:exp_order}

We found that vertex ordering has a significant impact on the efficiency of
the decomposition process, especially when we obtain a smaller kernel. We experimented with the following orderings.
%

%
\texttt{arbitrary} is the baseline strategy used in \cricca.
It keeps one arbitrary vertex from each block that is reduced and does not reorder the vertices.
\texttt{push\_front} is the one adopted as \SR 0.
It keeps one arbitrary vertex from each block that is reduced and moves the representative vertex to the front of the ordering.
\texttt{push\_back} does the opposite;
it keeps one arbitrary vertex from each block that is reduced and moves the representative vertex to the back.
\texttt{keep\_first} keeps the earliest vertex from each block that is reduced and does not reorder the vertices.

To clearly distinguish the effect of ordering (\SR 0) from the other \SRs, we ran \SCS (\SC with \SR 0) on seed 0 instances, where we experimented with each of the above orderings. We obtained results for both \KD and \KC. Figure~\ref{fig:exp-ordering} plots the log-scale distribution of running time of (\KD, \SCS) for the different ordering strategies compared to \texttt{arbitrary}.
\texttt{push\_front} gave the most speedup (10 to 100 times) and did not timeout on any instances. We believe this is because singleton vertices from reduced blocks tend to be non-neighbors which helps to quickly detect infeasible assignments. We observed similar behavior with the (\KC, \SCS) combination as shown in Figure~\ref{fig:exp-ordering-lp} in the Appendix.

\begin{figure}[]
    \centering
    \includegraphics[scale=0.25]{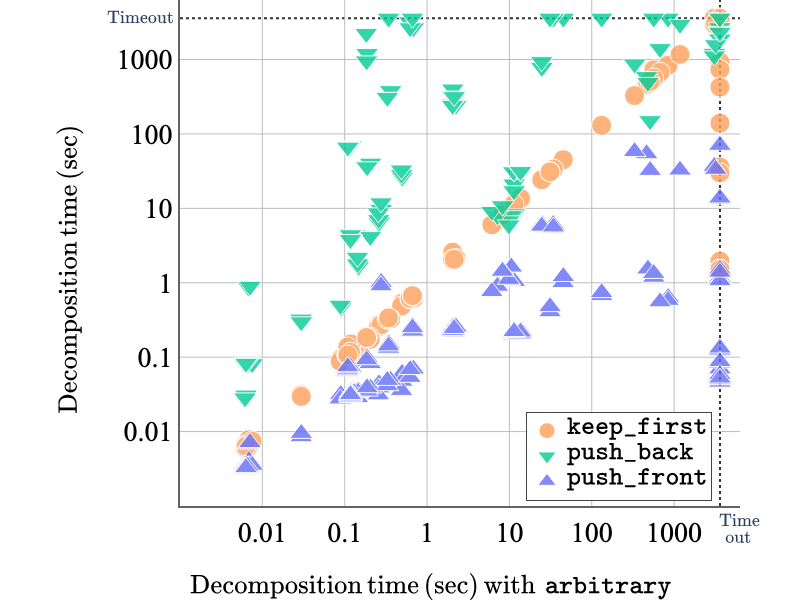}
    \captionof{figure}{%
        Log-log plot of decomposition runtime distribution on instances with seed 0
        using different vertex reordering strategies (and kernel $K_D$).
        Times relative to \texttt{arbitrary}.
    }
    \label{fig:exp-ordering}
\end{figure}






\subsection{Faster Decomposition} \label{sec:fastresults}

\begin{figure}[]
    \centering
    \includegraphics[width=0.4\textwidth]{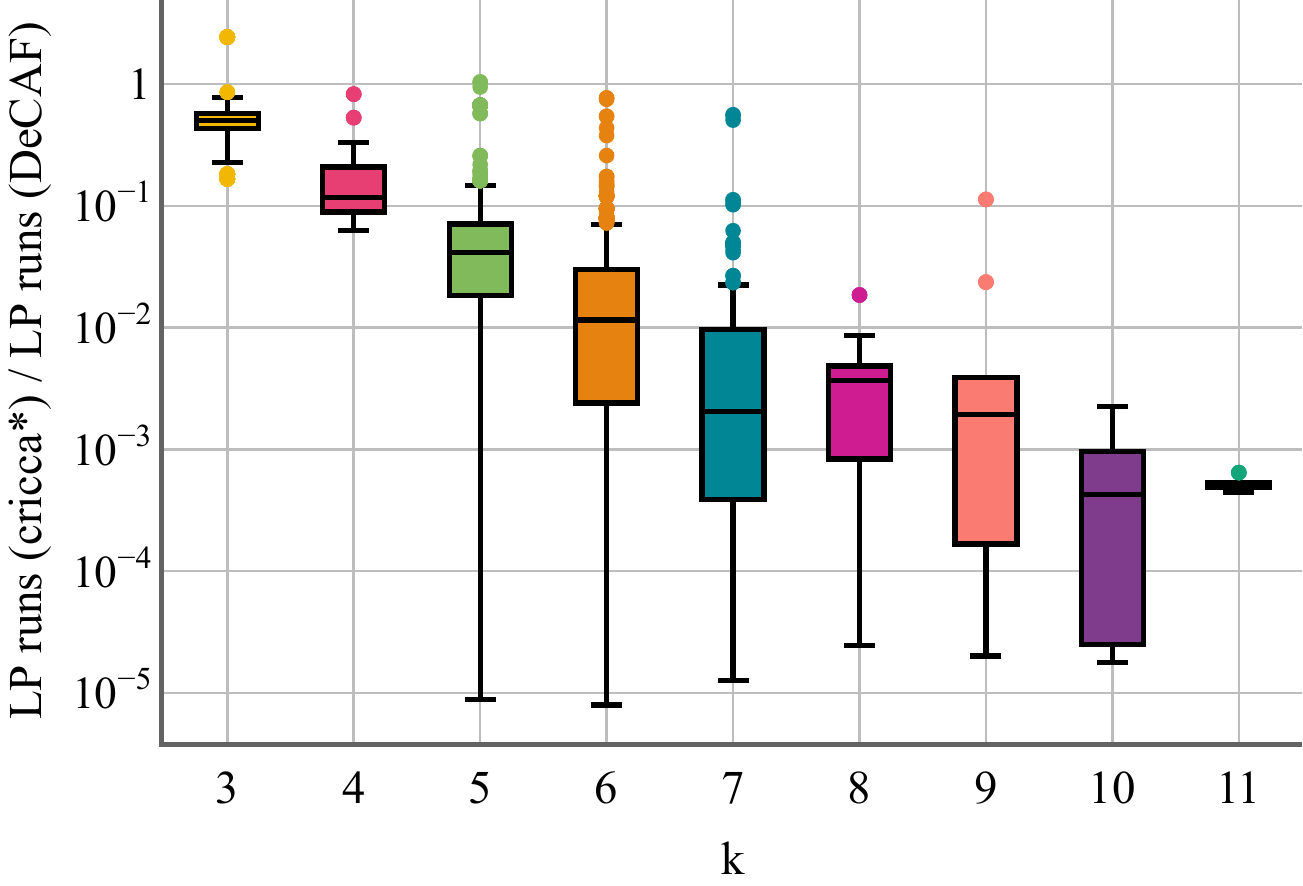}%
\caption{\decaf reduces the number of LP runs of \criccas by two orders of magnitude. We plot the ratio of number of LP runs on all instances that completed in 3600s with both algorithms. Instances are binned by $k$, the number of cliques in the desired decomposition, highlighting that the reduction in LP runs increases with the parameter value.}%
\label{fig:lpruns}%
\vspace{-1em}
\end{figure}

\begin{figure*}[t]
    \centering
    \includegraphics[scale=0.3]{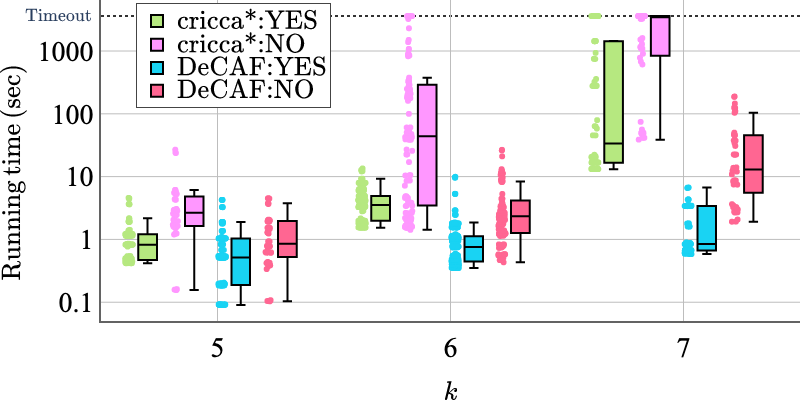}
    \includegraphics[scale=0.3]{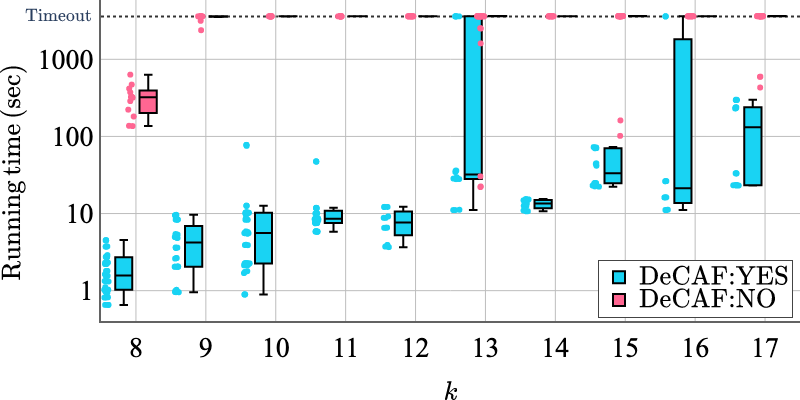}
    \captionof{figure}{%
        Runtime distribution (log-scale) for \criccas and \decaf on corpus of
        YES ($k_\text{in}=k$) and NO ($k_\text{in} = \lceil 0.8 k \rceil$) instances
        with seeds $[0,9]$ and $5 \leq k \leq 7$ (left), and \decaf only on $8 \leq k \leq 17$ (right).
    }
    \label{fig:exp-yesno}
    \vspace{-1em}
\end{figure*}
%

We show results comparing the time taken by \criccas and the time taken by \decaf for decomposing all instances reported in \cref{table:graphstats}. We also do an ablation study and report the effect of the smaller kernel on the running time, with and without search rules.

\textbf{Comparison with \criccas:}~\cref{fig:runtime} gives the ratio of the running time of \criccas and the running time of \decaf for different $k$ (lower is better). This includes the kernelization time as well as the time required for decomposition. For small $k$, we do not see a significant reduction in runtime but for larger $k$ we are able to obtain upto two orders of magnitude reduction in the running time, and the trend indicates that the reduction in runtime increases with $k$. Note that we do not include here information about instances for which either \criccas or \decaf timeout. In the entire corpus of 2556 instances, 297 instances timed out for \criccas and 3 for \decaf. More detailed information about these instances is in the Appendix (\cref{table:timeouts}). Note that the order in which the algorithm considers vertices to assign signatures to impacts the running time. In some cases, despite the smaller kernel and \texttt{push\_front} ordering, the running time can go up in the case of \decaf. Such instances have a ratio $>1$ in \cref{fig:runtime} and extremely rare.

\cref{fig:lpruns} shows the ratio of the number of times the LP solver was invoked by \criccas and the number of times the LP solver was invoked by \decaf. Because of the search rules and the smaller kernel, we are able to obtain upto three orders of magnitude reduction in the number of runs of the LP solver. Similar to running time, we see an increasing trend in the reduction in the search space with $k$.

\textbf{Ablation study:} To study the relative impact of the smaller kernel and the search rules on the running time of the decomposition algorithm, we tested all four combinations of the kernels (\KC and \KD) and decomposition algorithms (\SCS and \SD). Here, we only consider the running time of the decomposition algorithms and not the kernelization algorithms and do not include numbers for any instance that timed out across all four combinations. \cref{table:graphstats} shows the average reduction obtained in all 4 combinations for different $k$. The smaller kernel has a greater impact on the runtime compared to that of the search rules. Moreover, as expected, the average reduction increases as $k$ increases.

\subsection{Scalability}

Finally, we evaluate the overall scalability of our implementation.
In this section, we consider two input k-values,
(1) $k_\text{in}=k$ for YES instances and
(2) $k_\text{input}=\lceil 0.8 k \rceil$ for NO instances.
In the left figure in Figure~\ref{fig:exp-yesno}, we compare four configurations,
the combinations \criccas, \decaf
and two $k_\text{in}$ values for $k \in [5, 6, 7]$ with the instances of seeds $[0,9]$.
We observe that NO instances take longer than YES instances in both versions.
For \criccas, the median runtime reaches the time limit (3600 seconds) with $k=7$ for NO instances. In contrast, \decaf achieves a median of 13 seconds. For YES instances, for $k=7$ there is an order of magnitude difference in the runtimes of the two algorithms.

%

To test for scalability, we ran \decaf on 12 randomly-selected instances for each $k \in [8,17]$.
The right figure in Figure~\ref{fig:exp-yesno} shows the distribution of running time for the YES and NO instances.
The median runtime for NO instances with \decaf hits the time limit at $k=9$. Compare this with time taken by \criccas in the left figure of Figure~\ref{fig:exp-yesno}; it already hits the time limit at $k=7$. Since both \criccas and \decaf timeout for $k \geq 9$, we focus on the YES instances.

For YES instances, \decaf finished within 1000 seconds for all $k$
except for a few ``hard'' instances with $k=13, 16$
(This may happen because the worst-case running time is not polynomially bounded). Compare this with Figure 7 in ~\cite{cooley2021} which shows that \cricca already hits the time limit at $k=9$. We conclude that with \decaf we are able to scale to at least $1.5 \times$ larger $k$ than \cricca for YES instances. Details on the number of LP runs are shown in Figure~\ref{fig:exp-appendix-lp} (Appendix).

\notes{
Remaining 2.25 pages will be results.

\textbf{Prelimineries}
\begin{enumerate}
\item Give details of software, hardware and datasets.
\item Mention that since results for every graph cannot be individually shown, we have grouped the graphs by their k (this is what Madi's paper does). Mention that stats of the datasets (\#, n, m, k) are in Table 1.
\item Describe the experiment setup. What was the pipeline, what k values we tried.
\end{enumerate}

Table 1: \#, n, m, k, $n_{post}$, $k_{post}$, $n_{ker}$, time for Madi's kernel, time for our kernel, time for Madi's decomposition algorithm, time for our decomposition algorithm (for the min k for which instance was YES). All these for two datasets. Use data from Madi's paper that gives the size distribution of the graphs in the two datasets (or collect this ourselves if we are not using the same set of graphs). This will be a big table spanning both columns. Highlight the best timings.

\textbf{Smaller kernel}
\begin{enumerate}
    \item Show that kernel size is much smaller in our case compared to Madi's.
    \item figure: scatter plot, n vs $n_{ker}$, for 6 different values of k, different shade of a color for each k (similar to figure 3 in Madi's paper)
\end{enumerate}

\textbf{Faster decomposition}
\begin{enumerate}
    \item Show that performance rules and smaller kernel significantly reduce overall time (figure 1 in Intro and Table 1 above)
    \item Show that because the kernel is small, even without the performance tuning rules the algorithm is faster (ablation study). Table?
    \item Show that algorithm is fast on both YES as well as NO instances. figure: Similar to figure 9 in Madi's paper, a figure that shows how the complete pipeline performs on NO and YES instances (vary k). 3 different subplots for this that span the whole width of the page.

\end{enumerate}

\textbf{Effect of performance tuning (optional)} Talk about how many blocks were touched by performance tuning rules. Maybe give table.

\textbf{Space usage (optional)}
\begin{enumerate}
\item Small table showing that peak memory usage did not increase by much.
\end{enumerate}

\textbf{Scalability (optional)}
\begin{enumerate}
\item figure to show scalability. 2 subplots, arranged vertically in one column. One subplot: n vs $n_{ker}$, scatter plot (?),  for $k=5,10,15,20$, each k a different shade. Other plot: n vs $n_{ker}$, scatter plot (?),  for $k=5,10,15,20$, each k a different shade.
\end{enumerate}
}

\section{Conclusion}
\notes{
\begin{enumerate}
\item Conclude by summarising about the smaller kernel, rules for decomposition and faster decomposition algorithm.
\item Mention as future work the optimization version? Though this might sound like a shortcoming of our work in this paper (that we address the exact case when it is not very natural)
\end{enumerate}
}

We gave a $k2^k$ kernel for \EWCD, an exponential reduction over the best-known approach (a $4^k$ kernel).
Additionally, we describe new search rules that reduce the decomposition search space for the problem. Our algorithm \decaf combines these to achieve two orders of magnitude speedup over the state-of-the-art algorithm, \cricca; it reduces the number of LP runs by up to three orders of magnitude. Since our approach prunes away large parts of the search space, we are able to scale to solve instances with a larger number of ground-truth cliques ($k$) than previously possible, though the approach struggles to solve NO-instances in this setting. We believe additional rules for quickly detecting infeasibility will help the algorithm to achieve similar scalability when $k_{in} < k$.

\bibliographystyle{ACM-Reference-Format}
\bibliography{ewcc}

\newpage

\newpage
\appendix

\section{Appendix}

\subsection{Algorithms}

\begin{algorithm}[!h]
\begin{algorithmic}[1]
    \For {each non-null row $B_j$}\label{line:compfor}
    \If {$v^TWB_j\neq A_{ij}$}\label{line:compmulti}
        \Return false
        \EndIf
    \EndFor
    \If {$v^TWv\not\es A_{ii}$}
        \Return false
        \EndIf
    \State \Return true

\end{algorithmic}
\caption[]{\hspace*{-4.3pt}{.} \Comp($A$, $\widetilde{B}$, $W$, $i$, $v$)}\label{iWcompatible}
\label{alg:comp}
\end{algorithm}

\begin{algorithm}[!h]
\begin{algorithmic}[1]
    \State \textbf{let} $\cw_1,\cdots,\cw_k \geq 0$ be variables of the LP
    \For {all pairs of non-null rows $\basis_i,\basis_j$ s.t. $A_{ij}\neq \star$}
    \State Add LP constraint $\sum_{1\le q\le k}\basis_{iq}\basis_{jq}\cw_q = A_{ij}$
    \EndFor
    \If {the LP is infeasible}
        \Return {the null matrix}
    \Else {}
         \Return {the diagonal matrix given by $\cw_1,\ldots, \cw_k$}
    \EndIf
\end{algorithmic}
\caption[]{\hspace*{-4.3pt}{.} \InferCWLp($A$, $\widetilde{B}$)}\label{InferCliqWtsLp}
\label{alg:InferCWLp}
\end{algorithm}

\subsection{Proofs from Section 4}\label{app:proofs}

\begin{proof}[Proof of \cref{lem:correctness-rr2-yes}]
    We will prove that $B_u^TWB_v\es A_{uv}$ for all $u,v\in V(G)$, which guarantees validity. Let $D$ be a block reduced by \cref{rr:enhanced} and let $i$ be the representative vertex. There are 3 cases:
    \begin{itemize}
    \item $u,v\notin D$: Then $B_u^TWB_v={B'_{u}}^{T}W{B'_{v}}\es A'_{uv}=A_{uv}$.
    \item $u\in D,v\notin D$: Then $B_u^TWB_v={B'_u}^TWB'_i\es A'_{ui}=A_{ui}$.
    \item $u,v\in D$: If $A_{uv} = \star$ then this case is trivial. Assume $A_{uv}\neq \star$. Lemma 7 of ~\cite{feldmann2020} states that any two non-$\star$ entries in the same block of matrix $A$ must be equal. Hence, $B_u^TWB_v={B'_i}^TWB'_i\es A'_{ii}=A_{ij}= A_{uv}$.
    \end{itemize}
    This completes the proof that $(B,W)$ is a solution to $(A,k)$.
\end{proof}

\begin{proof}[Proof of \cref{lem:correctness-rr2-no}]
    Let $(B,W)$ be a solution of $(A,k)$.
    Since the block $D$ contains more than $k$ rows, by \cref{thm:main}, $D$ must be an identical block. Thus, there exist row indices $p,q \in D$ such that $B_p=B_q$.
    We define a solution $(B',W)$ for $(A',k)$ as $B'_u:=B_u$ for all $u\in V(G')\setminus\left\{ i \right\} $ and $B'_i:=B_p$, where $i$ is the representative vertex of $D$.

    To prove that $(B',W)$ is indeed a valid solution for $(A',k)$,
    it is sufficient to show that ${B'_u}^TWB'_v\es A'_{uv}$ for all $u,v\in V(G')$.  There are 3 cases to consider:
    \begin{itemize}
    \item $u,v\neq i$: Then ${B'_u}^TWB'_v={B_{u}}^{T}W{B_{v}}\es A_{uv}=A'_{uv}$.
    \item $u=i$,$v\neq i$: Then ${B'_i}^TWB_v={B_p}^TWB_v\es A_{pv}=A_{iv}=A'_{iv}$,
    where the second-to-last equality follows from $p$ and $i$ being in the same block $D$.
    \item $u=v=i$: In this case, ${B'_i}^TWB'_i={B_p}^TWB_p=B_p^TWB_q= A_{pq}\es A_{ij}= A'_{ii}$.
    As in the proof of \cref{lem:correctness-rr2-yes} the $\es$ here follows because any two entries (that are not $\star$) in the same block of matrix $A$ are equal (Lemma 7 of \cite{feldmann2020}). The third equality here is an equality (and not only a `$\es$ equivalence') as $A_{pq}$ is not a diagonal entry.
    \end{itemize}
    This completes the proof that if $(A,k)$ is a YES-instance then $(A',k)$ is also a YES-instance.
\end{proof}





\subsection{Additional Experimental Results}

This section includes supplemental information on instances that timed out, as well as the results of
vertex reordering strategies on $K_C$ and the number of LP runs on the corpus of instances from Section 6.4.

\textbf{Timeouts:} The number of timeouts with \criccas as well as \decaf for different $k$ are given in \cref{table:timeouts}.

\begin{table}[h]
\centering
\begin{adjustbox}{max width=0.47\textwidth}
\begin{tabular}{|l|r|r|r|r|r|r|r|r|r|}
\toprule
\multicolumn{1}{|l|}{}  & \multicolumn{1}{r|}{3}  & \multicolumn{1}{r|}{4} & \multicolumn{1}{r|}{5} & \multicolumn{1}{r|}{6} & \multicolumn{1}{r|}{7} & \multicolumn{1}{r|}{8} & \multicolumn{1}{r|}{9} & \multicolumn{1}{r|}{10} & \multicolumn{1}{r|}{11} \\
\midrule
\criccas & 0 & 0 & 21 & 46 & 132 & 20 & 30 & 27 & 21 \\
\decaf & 0 & 0  & 0 & 0 & 0 & 0 & 3  & 0 & 0 \\
\bottomrule
\end{tabular}
\end{adjustbox}
\caption{Table showing the number of instances in the main corpus that timed out for \criccas and \decaf for $3 \leq k \leq 11$.}
\label{table:timeouts}
\end{table}

\textbf{Impact of Vertex ReOrdering on $K_C$:} We also evaluated all four vertex ordering strategies using kernelized instances produced by the algorithm in ~\cite{cooley2021} using the instances with seed 0. While none of them significantly out-performed \texttt{arbitrary}, we note that \texttt{push\_front} also did not degrade performance.

\begin{figure}[ht]
    \centering
    \includegraphics[scale=0.3]{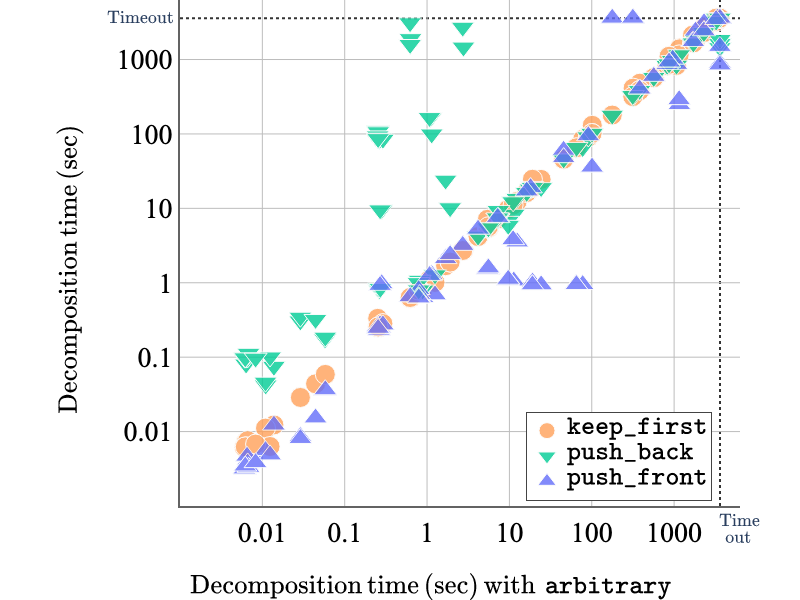}
    \captionof{figure}{%
        Log-scale plot of decomposition running time distribution with
        different vertex reordering strategies compared to \texttt{arbitrary}.
        Tested with the instances of seed 0, $k \in [3,11]$ and the original kernels ($K_C$).
    }
    \label{fig:exp-ordering-lp}
\end{figure}

\textbf{Distribution on LP runs when Scaling $k$:} Here, we include data on the number of executions of the LP solver on the corpus used in Section 6.4, broken out by YES- and NO-instances.

\begin{figure}[b]
    \centering
    \includegraphics[scale=0.3]{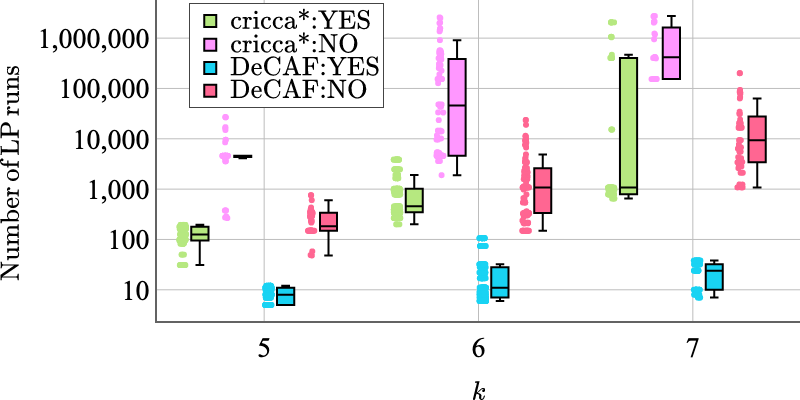}
    \includegraphics[scale=0.3]{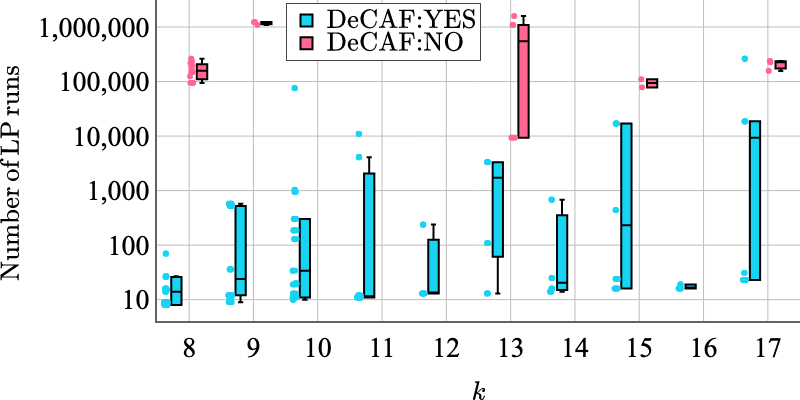}
    \captionof{figure}{%
        Distribution of number of LP runs (log-scale) for \criccas and \decaf on corpus of
        YES ($k_\text{in}=k$) and NO ($k_\text{in} = \lceil 0.8 k \rceil$) instances
        with seeds $[0,9]$ and $5 \leq k \leq 7$ (left), and \decaf only on $8 \leq k \leq 17$ (right).
    }
    \label{fig:exp-appendix-lp}
    \vspace{-1em}
\end{figure}

\end{document}